\documentclass[aps,prl,showpacs,twocolumn,nofootinbib,10pt, superscriptaddress]{revtex4-2}

\usepackage{graphicx}
\usepackage{float}
\usepackage{epstopdf}
\usepackage{amsmath}
\usepackage{amssymb}
\usepackage{mathrsfs}
\usepackage{amsthm}
\usepackage{bm}
\usepackage{url}
\usepackage[T1]{fontenc}
\usepackage{csquotes}
\usepackage{qcircuit}
\usepackage{xcolor}
\usepackage[utf8]{inputenc}
\MakeOuterQuote{"}

\newtheoremstyle{note}
  {\topsep/2}               
  {\topsep/2}               
  {}                      
  {\parindent}            
  {\itshape}              
  {.}                     
  {5pt plus 1pt minus 1pt}
  {}

\theoremstyle{note}
\newtheorem{theorem}{Theorem}

\newtheorem{proposition}{Proposition}

\theoremstyle{definition}

\theoremstyle{remark}

\newcommand{\scf}{\mathscr{F}_\mrm{sc}}
\newcommand{\scfp}{\mathscr{F}_\mrm{sc}^+}

\newcommand{\mrm}[1]{\mathrm{#1}}

\newcommand{\rmT}{\mathrm{T}}

\newcommand{\caH}{\mathcal{H}}

\newcommand{\gme}{\mathrm{GME}}

\newcommand{\be}{\begin{equation}}
\newcommand{\ee}{\end{equation}}
\newcommand{\ba}{\begin{align}}
\newcommand{\ea}{\end{align}}

\def\<{\langle}  
\def\>{\rangle}  
\newcommand{\ket}[1]{| #1\>}

\newcommand{\cref}[1]{Conjecture~\ref{#1}}
\newcommand{\Cref}[1]{Conjecture~\ref{#1}}

\makeatletter

\newcommand{\Rmnum}[1]{\expandafter\@slowromancap\romannumeral #1@}
\makeatother

\begin{document}
\title{Quantum Coherence: A Fundamental Resource for Establishing Genuine Multipartite Correlations}
\author{Zong Wang$^{\S}$}
\affiliation{School of Mathematical Sciences, MOE-LSC, Shanghai Jiao Tong University, Shanghai 200240,
China}
\affiliation{Shanghai Seres Information Technology Co., Ltd, Shanghai 200040, China}
\affiliation{Shenzhen Institute for Quantum Science and Engineering, Southern University of Science and Technology, Shenzhen 518055, China}
\email{These authors contribute equally to this work}
\author{Zhihua Guo$^{\S}$}
\affiliation{School of Mathematics and Statistics, Shaanxi Normal University, Xi'an 710119, China}
\email{These authors contribute equally to this work}
\author{Zhihua Chen$^{\S}$}
\affiliation{School of Science, Jimei University,  Xiamen 361021,
China}
\email{These authors contribute equally to this work}
\author{Ming Li}
\affiliation{College of the Science, China  University of Petroleum, Qingdao 266580, China}
\author{Zihang Zhou}
\affiliation{State Key Laboratory for Mesoscopic Physics, School of Physics,
Frontiers Science Center for Nano-optoelectronics, Peking University, Beijing 100871, China}
\author{Chengjie Zhang}
\email{zhangchengjie@nbu.edu.cn}
\affiliation{School of Physical Science and Technology, Ningbo University, Ningbo 315211, China}
\author{Shaoming Fei}
\email{smfei@mis.mpg.de}
\affiliation{School of Mathematical Sciences, Capital Normal University, Beijing 100048, China}
\affiliation{Max-Planck-Institute for Mathematics in the Science, Leipzig 04113, Germany}
\author{Zhihao Ma}
\email{mazhihaoquantum@126.com}
\affiliation{School of Mathematical Sciences, MOE-LSC, Shanghai Jiao Tong University, Shanghai 200240,
China}
\affiliation{Shanghai Seres Information Technology Co., Ltd, Shanghai 200040, China}
\affiliation{Shenzhen Institute for Quantum Science and Engineering, Southern University of Science and Technology, Shenzhen 518055, China}
\date{\today}
\begin{abstract}
We establish the profound equivalence between measures of genuine multipartite entanglement (GME) and their corresponding coherence measures. Initially we construct two distinct classes of measures for genuine multipartite entanglement utilizing real symmetric concave functions and the convex roof technique. We then demonstrate that all coherence measures for any qudit states, defined through the convex roof approach, are identical to our two classes of GME measures of the states combined with an incoherent ancilla under a unitary incoherent operation. This relationship implies that genuine multipartite entanglement can be generated from the coherence inherent in an initial state through the unitary incoherent operations.  Furthermore, we explore the interplay between coherence and other forms of genuine quantum correlations, specifically genuine multipartite steering and genuine multipartite nonlocality. In the instance of special three-qubit X-states (only nonzero elements of X-state are diagonal or antidiagonal when written in an orthonormal basis), we find that genuine multipartite steering and nonlocality are present if and only if  the coherence exists in the corresponding qubit states.
\end{abstract}
\maketitle
\section{Introduction}
Genuine multipartite entanglement (GME) plays a key role in many quantum information tasks, encompassing quantum teleportation \cite{RPMK,CG},quantum metrology \cite{Actmetro}, quantum key distribution \cite{RPMK}, quantum cryptography  \cite{RPMK,AK} and measurement-based quantum computing \cite{RPMK,RH}. Moreover, It also serves as a powerful tool for investigating critical physical phenomena such as black hole information paradox and quantum phase transition. As a cornerstone resource in the realm of quantum information, the detection of quantification of GME is of paramount importance. Many different GME measures have been proposed \cite{Geoent, YJ,SJ, Walter, Dai, Guo20, Schneeloch, Guo22, ZH, convexroofmea}. Despite the growing number of the subsystem, the detection of entanglement and practical computation of GME measures remain a significant and challenging endeavor in quantum information science \cite{Jungnitsch, GV, Huber, Qian18, Luo20, Yang22, Eltschka, Friis, DLXD}.

Besides the genuine multipartite entanglement, genuine multipartite non-locality (GMNL) and genuine multipartite steering (GMS), are all fundamental and important non-classical features of multipartite systems. Genuine multipartite non-locality have garnered substantial interest from the scientific community \cite{GS,DNSD,MG,JNNY}. The applications of GMNL spans a variety of quantum tasks, including the establishment of secure cryptographic protocols \cite{FGH,FGHD}, the reduction of communication complexity in information transmission \cite{HRSR} and the generation of device-independent random numbers \cite{SASA}. Despite the inherent challenges in detecting and characterizing GMNL, the development of multipartite Bell-type inequalities has provided valuable tools for its verification \cite{MG,QSCC,MAIT}. Genuine multipartite steering (GMS), first introduced in \cite{QM},  holds profound significance in various quantum information processes, including quantum teleportation \cite{M}, quantum secret sharing  \cite{IYQG} and quantum subchannel discrimination  \cite{KXY}. Unlike GME and GMNL, GMS is not symmetric intrinsically, which complicates the detection and quantification of this correlation \cite{DPG,KXJ,JFMN,LSL, ZRS}.

Quantum coherence, rooted in the superposition principle, is the cornerstone of quantum mechanics. Unlike classical state (that is, incoherent state), a quantum state (e.g.,Schr$\ddot{o}$dinger's cat) is in a superposition of different states(the cat is both dead and alive at the same time), which leads to the quantumness of the state, and brings quantum advantages. So it is not surprising that quantum coherence can be served as a vital physical resource across a wide range of fields, including 
quantum metrology \cite{VSL,BRL,IRWS}, and quantum computation et. al \cite{PW,MH,JBD,JDN}. A key issue when studying quantum coherence is how to quantify quantum coherence. With the development of resource theory, a number of coherence measures have been put forward \cite{BCP,SAP17,HU, XianShi, Speedlim, SD}. Actually, many of which draw inspiration from results of entanglement theory, such as coherence cost(also called coherence of formation) \cite{XHZ,ADY}, robustness of coherence \cite{CTM} and distillable coherence (also called the relative entropy of coherence) \cite{BCP,ADY}. This has led to a surge of interest in the interconversion between entanglement and coherence under particular conditions \cite{JJH,WJ,AUH,SC,NFM,XTF,Kim}. It has been demonstrated that any degree of coherence in a given reference basis can be transformed into bipartite entanglement through incoherent operations, as shown in \cite{JJH, AUH, Kim}. Furthermore, a comprehensive operational one-to-one correspondence between coherence and entanglement measures for bipartite quantum systems has been established \cite{HZ}.

In contrast to the well-explored relationship between quantum entanglement and coherence, the complex interconnections between quantum coherence and other types of quantum correlations have not been thoroughly examined.  The correlations between quantum coherence and quantum nonlocality, as well as genuine quantum entanglement and genuine quantum nonlocality, have been explored in \cite{Regula,XIYA}. The emergence of quantum discord via multipartite, incoherent operations is intrinsically bounded by the consumption of quantum coherence within the subsystems \cite{JBD}. The dynamics between quantum coherence and additional convex resources, including quantum Fisher information and the so-called "magic" resources, have been clarified \cite{TANKC, Mukhopadhyay}. Nevertheless, the investigation into the interplay between quantum coherence and genuine multipartite quantum correlations is still quite limited, and as a result, our comprehension of this intricate relationship remains largely incomplete.
\begin{figure}
	\centering
	\includegraphics[width=0.7\linewidth]{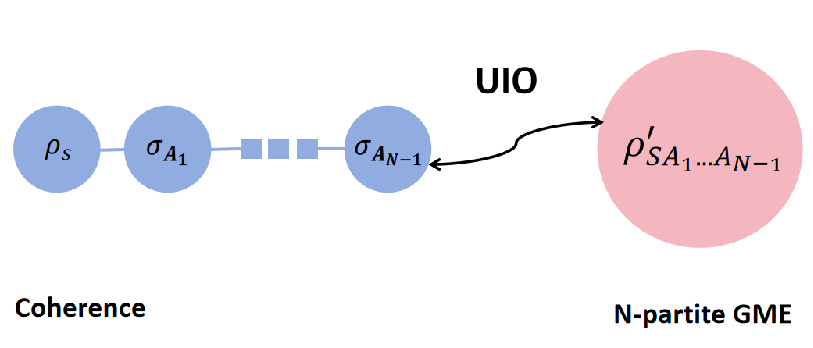}
	\caption{ Connection between GME and coherence. All coherence measures defined by the real symmetric concave functions for pure states with extension to mixed states by convex roof are equivalent to the GME measures of the $N$-partite states $\rho_{SA_{1}...A_{N-1}}$ generated by a unitary incoherent operation (UIO) on the initial states $\rho_{S}$ and $N-1$ incoherent ancilla states $\rho_{A_{1}},...,\rho_{A_{N-1}}$.}
	\label{Fig.1}
\end{figure}

In this study, we develop two distinct measures of genuine multipartite entanglement (GME) and establish a link between GME measures and coherence measures. We demonstrate that all coherence measures defined by the symmetric concave functions for pure states, and generalized to mixed states through the convex roof construction, are equivalent to our two corresponding GME measures for the states produced by a unitary incoherent operation performed on the initial states with an incoherent ancilla, as depicted in FIG. \ref{Fig.1}.
Moreover, we explore the operational connections between coherence and other forms of quantum correlations, such as genuine multipartite steering (GMS) and genuine multipartite nonlocality (GMNL).
Our results enhance the comprehension of the interconversion processes between genuine multipartite quantum correlations and quantum coherence. The established linkage allows for the extrapolation of numerous findings from the genuine multipartite correlations to the field of quantum coherence, and vice versa, thereby broadening the potential applicability of these concepts.

The structure of this paper is as follows. Section \Rmnum{1} provides an introduction. In Section \Rmnum{2}, we introduce two types of measures for genuine multipartite entanglement (GME) and offer a comparative analysis of these measures. Section \Rmnum{3} is dedicated to establishing an operational link between GME and coherence monotones constructed using the convex roof method. In Section \Rmnum{4}, we delineate the operational connections between coherence and other forms of genuine quantum correlations. Finally, Section \Rmnum{5} summarizes the key findings and contributions of this paper.

\section{Quantum Coherence Monotones And Genuine Multipartite Entanglement Monotones}

Consider the $d$-dimensional Hilbert space $\caH^d$ with an orthonormal basis $\{|j\>\}_{j=1}^{d}$.
The coherence vector of a pure state
$\ket{\psi}=\sum_{i=0}^{d-1}c_{i}\ket{i}$ in $\caH^d$
is defined as \cite{HZ,SD}:
\begin{equation}
\mu(\ket{\psi})= (|c_{0}|^{2}, |c_{1}|^{2},..., |c_{d-1}|^{2})^{\rmT}.
\end{equation}
Denote by $\scf$ the set of symmetric concave functions on the probability simplex.
According to \cite{HZ,SD}, quantum coherence monotones are in one-to-one correspondence with symmetric concave functions on the probability simplex.
Let $\scfp$ be the subset of functions $f$ that satisfy the additional property: $f(p)=0$ if and only if $\max p_j=1$. Note that this condition automatically guarantees that $f(p)$ is nonnegative.
Given any $f\in \scf$, a coherence monotone for a pure state $|\psi\>$ can be constructed as follows \cite{HZ,SD},
\begin{equation}
C_f(\ket{\psi})= f(\mu(\ket{\psi})).
\end{equation}
The monotone for the mixed state $\rho$ is  defined by the convex roof,
\begin{equation}\label {eq:Cfmix}
C_f(\rho)=\min_{\{p_j, |\psi_j\rangle\}} \sum_j p_j C_f(|\psi_j\rangle),
\end{equation}
where the minimum is taken over all pure state decompositions of $\rho$ as $\rho=\sum\limits_j p_j |\psi_j\rangle\langle \psi_j|$.

Genuine multipartite entanglement (GME) measures are defined to quantify the entanglement which exists in a quantum system across all of its subsystems simultaneously.
An $N$-partite pure state $|\psi\rangle\in H_1\otimes H_2\otimes \cdots\otimes H_N$ is called as biseparable if it can be written as $|\psi\rangle=|\psi_{\gamma}\rangle\otimes|\psi_{\bar{\gamma}}\rangle$,
where $|\psi_{\gamma}\rangle\in H_{k_{1}}\otimes H_{k_{2}}\otimes \cdots\otimes H_{k_{l}}$,  $|\psi_{\bar{\gamma}}\rangle\in H_{k_{l+1}}\otimes H_{k_{l+2}}\otimes \cdots\otimes H_{k_{N}}$.
$\gamma={k_{1},k_{2}\cdots,k_{l}}$ is an arbitrary nonempty subset of $\mathcal{N}=\{1,2,\cdots,N\}$ and $\bar{\gamma}$ the corresponding complementary set with respect to $\mathcal{N}$ (for example, when $N=3$, $1|23$ is a bipartition of $\{1,2,3\}$, $2|13$ is also a bipartition of $\{1,2,3\}$). And an $N$-partite mixed state $\rho$ is  biseparable, if it can be expressed as a convex combination of biseparable pure states, i.e., $\rho=\sum\limits_j p_j |\psi_j\rangle\langle \psi_j|$ with  $|\psi_j\rangle$ being biseparable state. It is worth noting that the $|\psi_j\rangle$ can be biseparable with respect to different bipartitions for different $j$. If a quantum state is not biseparable, then it is genuine multipartite entangled (GME).
Different genuine multipartite entanglement measures have been proposed to quantify the degree of entanglement in genuinely multipartite entangled states.
A function $E: H_1\otimes H_2\otimes \cdots \otimes H_N\rightarrow R$ is a GME measure if it satisfies the following conditions:

 (C1) $E(\rho)=0$ if and only if $\rho$ is biseparable state. 

 (C2) Convexity
 \begin{equation*}
 E(\sum_{j}p_{j}|\psi_{j}\rangle\langle \psi_j|)\leq\sum_{j}p_{j}E(|\psi_{j}\rangle\langle\psi_j|).
 \end{equation*}

 (C3) Monotonicity under nonselective local operations and classical communications (LOCC) $\Lambda$, namely $E(\rho)$ is nonincreasing under LOCC.
 \begin{equation*}
 E(\Lambda(\rho))\leq E(\rho).
 \end{equation*}

 (C4) Monotonicity under selective LOCC, namely $E(\rho)$  is nonincreased on average
under stochastic LOCC,
 \begin{equation*}
 \sum_{j}p_{j}E(\rho_{j})\leq E(\rho).
 \end{equation*}
Here LOCC is stochastic in
the sense that $\rho$ can be converted to $\rho_j$ with some probability
$p_j$.

A multipartite entanglement monotone satisfies conditions (C2-C4),  whereas a GME measure fulfills all four of the aforementioned conditions.
Now we present two approaches for constructing the GME measures based on symmetric concave functions on the probability simplex.
The GME measure constructed by the first approach is called min-GME-measure, while the one constructed by the second approach is called geo-GME-measure.

Consider an $N$-partite pure state $|\psi\rangle\in H_1\otimes H_2\otimes \cdots\otimes H_N$ and a real symmetric concave function $f\in \scf$.
For any nonempty subset  $\gamma$  of $\mathcal{N}=\{1,2,\cdots,N\}$ and the corresponding complementary set $\bar{\gamma}$ with respect to $\mathcal{N}$, we have the bipartition $\gamma|\bar{\gamma}.$
The first approach is as follows. The entanglement measure of $|\psi\rangle$ with respect to the bipartition $\gamma|\bar{\gamma}$ can be defined as follows,
\begin{equation}\label{eq:Efgamma}
E^{f}_{\gamma}(|\psi\rangle):=f(\lambda_{\gamma}(|\psi\rangle)).
\end{equation}
where $\lambda_{\gamma}(|\psi\rangle)$ represents the Schmidt vector of $|\psi\rangle$ under bipartition $\gamma|\bar{\gamma}$.
The min-GME-measure of $|\psi\rangle$, which is denoted as $E^{f}_\gme(|\psi\rangle)$, is defined as  the minimum of $E^{f}_{\gamma}(|\psi\rangle)$ over all bipartitions of $\mathcal{N}$ as follows,
\begin{equation}\label{eq:Efmingme}
E^{f}_\gme(|\psi\rangle):=\min\limits_{\gamma|\bar{\gamma}}E^{f}_{\gamma}(|\psi\rangle),
\end{equation}

The min-GME-measure of any mixed state $\rho \in H_1\otimes H_2\otimes \cdots\otimes H_N$ is defined by the convex roof construction as,
\begin{equation}\label{eq:EfmingmeMix}
E^{f}_\gme(\rho):=\min_{\{p_j, |\psi_j\rangle\}}\sum_{j}p_jE^{f}_\gme(|\psi_j\rangle)
\end{equation}
where the minimum is taken over all pure-state decompositions  $\rho=\sum\limits_j p_j|\psi_j\rangle\langle\psi_j|$. If $E^{f}_\gme$ in the above equation is replaced by $E^f_\gamma$, then we get an entanglement measure with respect to the partition $\gamma$ \cite{GV, ZH}.
\begin{equation}\label{eq:EfgammaMix}
E^{f}_\gamma(\rho):=\min_{\{p_j, |\psi_j\rangle\}}\sum_{j}p_jE^{f}_\gamma(|\psi_j\rangle).
\end{equation}
Due to the above definitions we have $E^{f}_\gme(\rho)\leq E^{f}_\gamma(\rho)$ for any subset $\gamma$ and any state $\rho$.

Another GME measure  geo-GME-measure is defined as the geometric mean of $E^{f}_{\gamma}(|\psi\rangle)$ over all bipartitions $\gamma|\bar{\gamma}$ for any pure state $|\psi\rangle$,
\begin{equation}\label{eq:Efgeogmep}
G^{f}_\gme(|\psi\rangle):=[\prod_{\gamma}E^{f}_{\gamma}(|\psi\rangle)]^{\frac{1}{c(\alpha)}}
\end{equation}
where
\begin{equation*}
c(\alpha)=
\begin{cases}
\sum_{m=1}^{\frac{N-1}{2}}\tbinom{N}{m}, $if \textit{N} is odd$,\\
\sum_{m=1}^{\frac{N-2}{2}}\tbinom{N}{m}+\frac{1}{2}\tbinom{N}{\frac{N}{2}}, $if \textit{N} is even$.\\
\end{cases}
\end{equation*}
The geo-GME-measure can be generalized  to mixed state $\rho$ via a convex roof construction, i.e.
\begin{equation}\label{eq:Efgeogme}
G^{f}_\gme(\rho):=\min_{\{p_j, |\psi_j\rangle\}}\sum_{j}p_jG^{f}_\gme(|\psi_j\rangle),
\end{equation}
here the minimum is taken over all pure-state decompositions  $\rho=\sum\limits_j p_j|\psi_j\rangle\langle\psi_j|$.

\begin{theorem}
For any  $f\in\scf$, $E^{f}_\gme$ and $G^{f}_\gme$ are both entanglement monotones. If in addition $f\in\scfp$, then $E^{f}_\gme$ and $G^{f}_\gme$ are both GME measures.
\end{theorem}

The proof is in Appendix A.
Then we present the relationship of these two types of GME measures.
\begin{theorem}
For an arbitrary $N$-partite mixed state $\rho$, one can prove the following inequality,
\begin{equation}\label{}
  G_{\mathrm{GME}}^f(\rho)\geq E_{\mathrm{GME}}^f(\rho).
\end{equation}
\end{theorem}

One can see detailed proof in Appendix B.

As illustrative instances of GME measures, consider the function $f(p)=\sqrt{2\sum_{i\neq j}p_{i}p_{j}}$. In this case $E^{f}_\gme(\rho)$ turns into  the GME concurrence \cite{ZH}, and $G^{f}_\gme(\rho)$ becomes the geometric mean of bipartite concurrence (GBC) \cite{YJ} multiplied by a constant. The corresponding function is given by $f(p)=\sqrt{\frac{d_{min}}{d_{min}-1}\sum_{i\neq j}p_{i}p_{j}}$, where $d_{min}$ denotes the dimension of the smaller subsystem under bipartition $\gamma|\bar{\gamma}$.
When the function is defined as $f(p)=-\sum_{j}p_{j}\log p_{j}$,  $E^{f}_\gme(\rho)$ represents the multipartite counterpart of the entanglement of formation and $G^{f}_\gme(\rho)$ is the geometric mean of entanglement of formation.

Worthy of special note, in addition to fulfilling the aforementioned properties (C1-C5), the geometric genuine multipartite entanglement  (geo-GME) measure is also characterized by its smoothness. Notably, the geo-GME measure does not exhibit sharp peaks when assessing continuously varying pure states. The detail is thoroughly explored in \cite{YJ}.

\section{Operational Relation Between GME And Coherence}

Hereafter, we explore the interconnection between GME and coherence through the construction of a specific type of unitary incoherent operation (UIO), which we denote as
$U$. Consider initial state $\rho\otimes |0\rangle\langle 0|\otimes\cdots\otimes |0\rangle\langle 0|$ with
$\rho$ in $H^{d}$ and $N-1$ ancilla quantum states $|0\rangle\langle 0|\otimes\cdots\otimes |0\rangle\langle 0|$ in $ H^{d_{1}}\otimes...\otimes H^{d_{N-1}}$ satisfying $d\leq d_{i}, 1\leq i\leq N-1$.
We employ UIO as the converting unitary operator
\begin{equation*}
U=\sum_{i=0}^{d-1}|i\rangle\langle i|\otimes \sigma_{i}^{\otimes N-1},
\end{equation*}
where
\begin{equation}\label{7}
\sigma_{i}=
\begin{cases}
\mathbb{I},&\mbox{if $i=0$},\\
P_{\pi(i)},&\mbox{if $i\neq0$},\\
\end{cases}
\end{equation}
where $\mathbb{I}$ is the identity operator and $P_{\pi(i)}$ the permutation operator that transforms $|0\rangle$ into $|i\rangle$, that is,

\begin{equation*}
P_{\pi(i)}|0\rangle\langle 0|P_{\pi(i)}^{T}=|i\rangle\langle i|.
\end{equation*}

In the theorem below we prove that the operator UIO, namely $U$, can fully convert the coherence of initial state $\rho$ into genuine multipartite entanglement {of the $N$-partite state $\rho'$ via the following process}
\begin{equation*}
\rho'=U(\rho\otimes|0\rangle\langle 0|^{\otimes N-1})U^{\dag}.
\end{equation*}
Note that this operation is incoherent, which implies no additional coherence is involved during the converting process.

\begin{theorem}
For any $f\in\scfp$, we have
\begin{equation*}
E^{f}_{GME}(\rho')=G^{f}_{GME}(\rho')=C^{f}(\rho).
\end{equation*}
\end{theorem}

The detailed proof is in Appendix C.



It is worthy to be mentioned, using Theorem 3, a lot of noval GME measures can be derived from coherence measures, and vice versa, revealing profound connections between multipartite entanglement and coherence.
For example, we consider $N$-qubit case, $U_{Nq}=|0\rangle\langle 0|\otimes I_{2}^{\otimes N-1}+|1\rangle\langle 1|\otimes \sigma_{x}^{\otimes N-1}$, where $I_{2}$ and $\sigma_{x}$ stand for two dimensional identity matrix and Pauli $X$ matrix. The quantum circuit of $U_{Nq}$ is described in FIG. \ref{Fig.2}. Then we can get the connection between GME-concurrence and $l_1$-norm coherence measure, shown in the following proposition.

\begin{proposition}
For any $N$-qubit (mix or pure) state $\rho$, the GME-concurrence of $\rho'$ is equal to the $l_1$-norm coherence measure of $\rho$.
\end{proposition}

\begin{proof}
The $l_1$-norm coherence measure of $\rho$ is defined as the sum of the modulus of the nondiagonal elements of $\rho$, i.e., $C_{l_1}(\rho)=\sum_{i\neq j}|\rho_{ij}|$, where $\rho_{ij}$ represents the entries of $\rho$. 
Inspired by the method of \cite{HZ}, first we consider pure state, let  $|\Psi\rangle=(\alpha|0\rangle+\beta|1\rangle)|0^{\otimes (N-1)}\rangle$, where $|0\rangle^{\otimes (N-1)}$ is $N-1$ ancilla qubits and $|\alpha|^{2}+|\beta|^{2}=1$. Actually, $|\Psi^{'}\rangle=U_{Nq}|\Psi\rangle=\alpha|0\rangle^{\otimes (N)}+\beta|1\rangle^{\otimes (N)}$ is an X-state and its GME-concurrence is $2|\alpha\beta|$ \cite{SMH}, which is equal to the $l_1$-norm coherence measure of $|\Psi\rangle$.

For mixed state $\rho_{m}=\sum_{i,j=0}^{1}\rho_{ij}|i\rangle\langle j|$, we can use the method similar to the case of pure states, let
\begin{eqnarray*}
\rho_{m}^{in}=(\sum_{i,j=0}^{1}\rho_{ij}|i\rangle\langle j|)\otimes (|0\rangle\langle0|)^{\otimes (N-1)}.
\end{eqnarray*}
Obviously, $\rho^{'}=U_{Nq}\rho_{m}^{in} U_{Nq}^{\dag}=\sum_{i,j=0}^{1}\rho_{ij}|ii\cdots i\rangle\langle jj \cdots j|$ is a $2^{N}$-dimensional X-state. Due to the result in \cite{ZH, SMH}, we can calculate the GME-concurrence of $\rho^{'}$ is $\sum_{i\neq j}|\rho_{ij}|$, which is equal to the $l_1$-norm coherence of $\rho$.
\end{proof}


Recently, a new GME measure named as geometric mean of bipartite concurrence (GBC) was introduced in \cite{YJ}. Besides satisfying the conditions (C1-C4), GBC also has some good properties such as smoothness (which means that GBC cannot make sharp peaks under continuous measurements of variable pure states) as described in Section \Rmnum{2}. Based on the results in Theorem 3 and Proposition 1, we can get the following proposition.
\begin{proposition}
For all $N$-qubit (mixed or pure) state $\rho$, GBC of $\rho'$ is equivalent to the $l_{1}$-norm coherence of $\rho$ in terms of a constant factor.
\end{proposition}



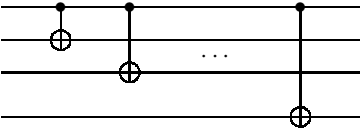
\begin{figure}
  \centerline{
  \Qcircuit @C=2em @R=.5em {& \ctrl{1} & \ctrl{2} & \cds{4}{\cdots} & \ctrl{4} & \qw\\& \targ & \qw &\qw & \qw & \qw \\& \qw & \targ & \qw & \qw & \qw \\
  & & & & & \\& \qw & \qw & \qw & \targ & \qw \\}
  }
  \caption{The quantum circuit of $U_{Nq}$}\label{Fig.2}
\end{figure}



\section{Connection Between Coherence And Genuine Quantum Correlations}
A system state can be transformed into an $N$-partite state which
has the form $\rho^{'}=\sum_{i,j}s_{ij}|ii...i\rangle\langle jj...j|$ with $\sum_{i}s_{ii}=1$  via the process introduced in Section \Rmnum{3}, and its coherence is equal to the GME of the corresponding state $\rho'$. Here we will explore the connection between the coherence of the system state and other genuine quantum correlations of the special X-state, which include the genuine multipartite steering and genuine multipartite nonlocality.
Let us consider any qubit state,
\begin{equation}\label{eq:EfSystemState}
	\rho^{I}=p|0\rangle\langle 0|+r|0\rangle\langle 1|+r|1\rangle\langle 0|+(1-p)|1\rangle\langle 1|,
\end{equation}
where $p\in[0,1]$ and $r=|r|e^{i\theta}$ satisfying $|r|^{2}\leq p(1-p)$ ensuring that the density matrix
$\rho^I$ remains semi-positive definite.
By employing the local unitary operations $[1,0;0,e^{\rm{i}\theta}]$ on $\rho^I,$ we can obtain a real quantum state as
\begin{equation}\label{eq:EfSystemState}
\rho=p|0\rangle\langle 0|+|r||0\rangle\langle 1|+|r||1\rangle\langle 0|+(1-p)|1\rangle\langle 1|,
\end{equation}

A class of three-qubit X-state can be obtained by performing the converting incoherent operation $U$ in Sec. \Rmnum{3} on { the coupled system} $\rho\otimes |0\rangle\langle 0|^{\otimes 2}$
\begin{equation}\label{eq:EfGHZType}
\begin{aligned}
\rho'=&p|000\rangle\langle 000|+|r|(|000\rangle\langle 111|+|111\rangle\langle 000|) \\
&+(1-p)|111\rangle\langle 111|,
\end{aligned}
\end{equation}

It is known that an $N$-partite state $\rho_{A_1\cdots A_N}$ is genuine multipartite nonlocal if there exsit measurements $M_{x_i}$ performed on the parties $A_i$ with $x_i\in\{0,1\}$ and outcomes $c_i\in \{0,1\}$ such that the
probabilities $p(c_1\cdots c_N|M_{x_1}\cdots M_{x_N})=\rm{Tr}[\mathit{\rho_{A_1\cdots A_N}\bigotimes\limits_{i=1}^N M_{x_i}^{c_i}}]$ do not admit the hybrid local hidden variable (LHV) model as follows
\cite{GS,DNSD,MG,JNNY}
\begin{equation}\label{eq:hlhv}
\begin{aligned}
&&p(c_1\cdots c_N|M_{x_1}\cdots M_{x_N})~~~~~~~~~~~~~~\\ \nonumber
&&=\sum\limits_{\alpha|\bar{\alpha}}\sum\limits_{\lambda} p(\lambda)p(c_{\alpha}|M_{x_{\alpha}},\lambda)p(c_{\bar{\alpha}}|M_{x_{\bar{\alpha}}},\lambda),
\end{aligned}
\end{equation}
where  $\alpha=\{i_1,\cdots, i_k\}$ is a nonempty subset of $I=\{1,2,\cdots,N\}$ and $\bar{\alpha}=I\setminus \alpha$ ,  $c_{\alpha}=\{c_{i_1}\cdots c_{i_k}\},$ and $c_{\bar{\alpha}}=\{c_{1}\cdots c_{N}\}\setminus C_{\alpha},$ $M_{x_{\alpha}}=\{M_{x_{i_1}},\cdots, M_{x_{i_k}}\}$ and $M_{x_{\bar{\alpha}}}=\{M_{x_{1}},\cdots, M_{x_{N}}\}\setminus M_{\alpha}$.
In \cite{QSCC}, a Bell-type inequality has been introduced, the violation of which signifies the presence of genuine multipartite nonlocal correlations,
\begin{equation}\label{eq:hardytest}
\begin{aligned}
&p(0_{I}|M_0,\cdots, M_0)\\ \nonumber
&-\sum_{k\in I\setminus \{k'\}}p(1_{k'}1_{k}0_{\overline{k'k}}|M_{x_{k'}=1}M_{x_k=1}M_{x_{\overline{k'k}}=0}) \\
&-\sum_{k\in I}p(0_{I}|M_{x_k=0}M_{x_{\overline{k}}=1})
\leq 0,
\end{aligned}
\end{equation}
where the $k$-th local observer measures two alternative observables ${M_0,M_1}$ with two outcomes labeled by ${0,1}$, $\overline{k}=I\setminus\{k\}, \overline{k'k}=I\setminus \{k,k'\}$, and $k'\in I$ is fixed as one.
For example, in three-qubit case, to violate the above inequality, one need to find observables  $M_{x_i}$ for each particle $i$ such that
\begin{equation}\label{eq:EfHardy2}
\begin{aligned}
&&H=p(000|M_0M_0M_0)-p(100|M_1M_0M_0)\\
&&-p(010|M_0M_1M_0)-p(001|M_0M_0M_1)\\&&-p(110|M_1M_1M_0)-p(101|M_1M_0M_1)>0,
\end{aligned}
\end{equation}
Eq. (\ref{eq:EfHardy2}) indicates the genuine multipartite nonlocality for three-qubit states.

Based on the aforementioned approach, we deduce that the maximum value of  $H$ in Eq. (\ref{eq:EfHardy2}) is positive for the state expressed in Eq.(\ref{eq:EfGHZType}) if and only if $|r|>0$. The detailed calculation is provided in Appendix D. The findings are visually represented in Fig. \ref{Fig.3}.


\begin{figure}[htbp]
	\centering
	\centering
	\includegraphics[width=4cm,height=3cm]{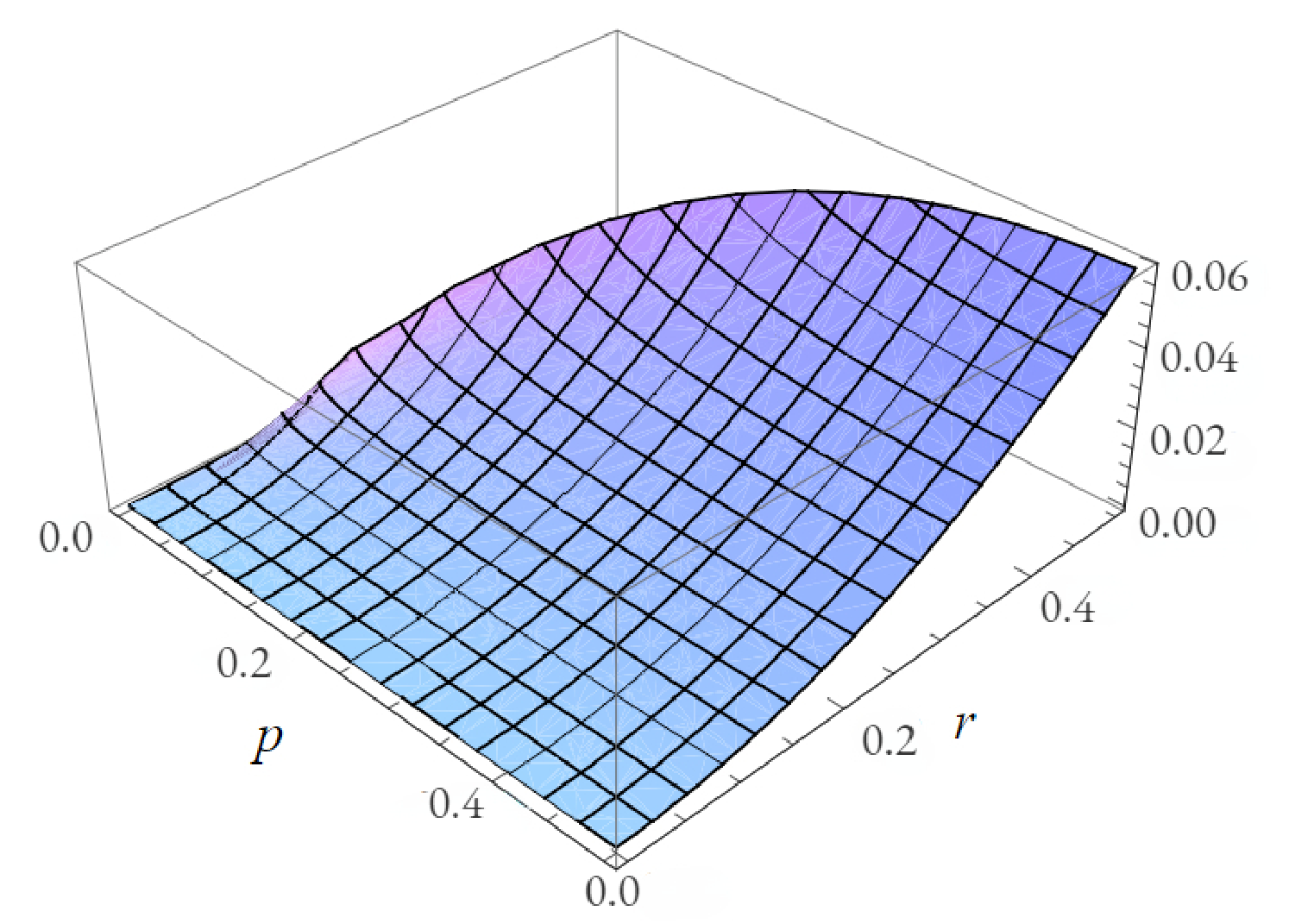}
	\includegraphics[width=4cm,height=3cm]{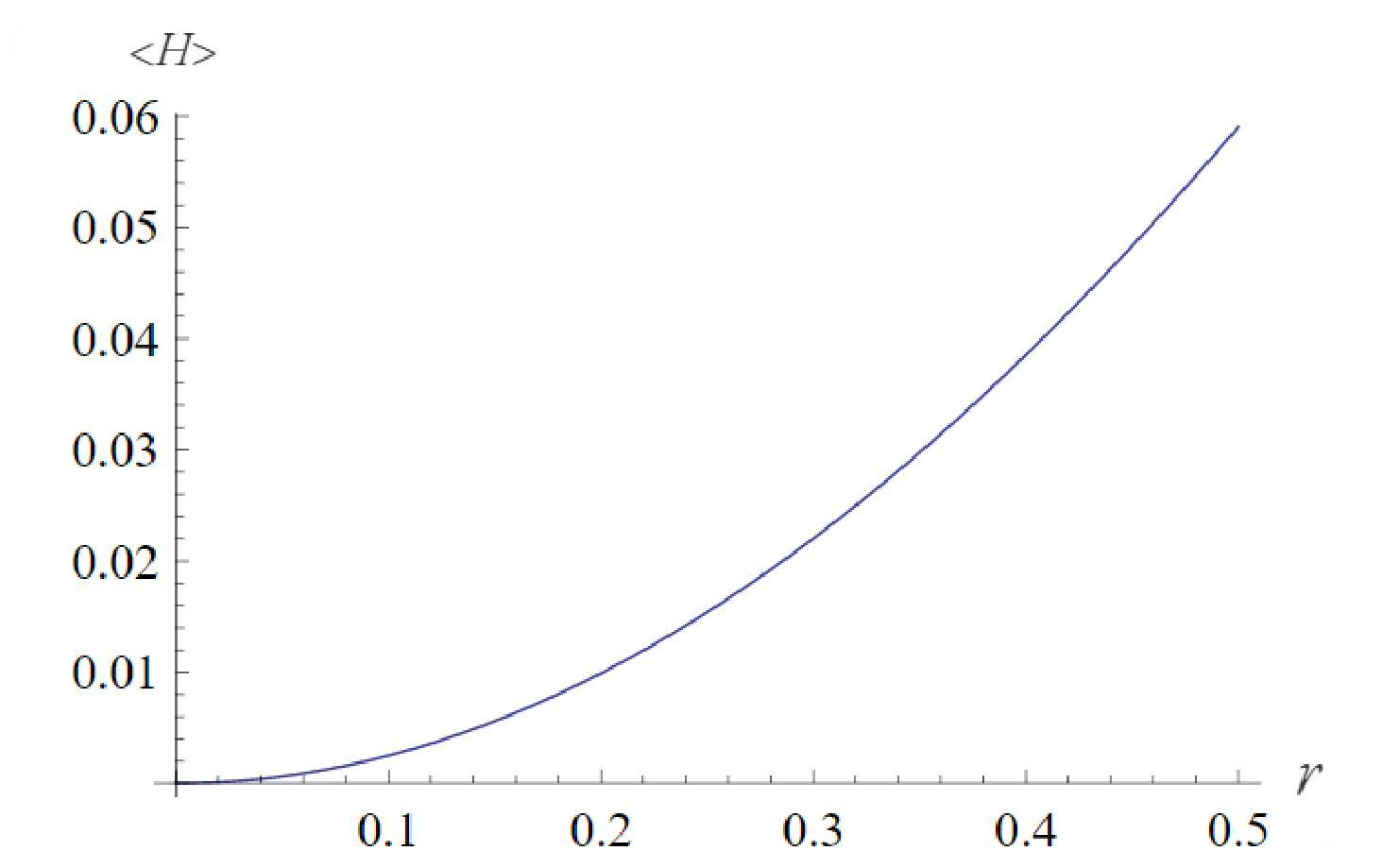}
	\caption{The left panel illustrates the maximum value of $H$ as a function of $p$ and $r$
for the parameter range $0\leq p\leq0.5$ and $0\leq |r|\leq\sqrt{p(1-p)}$. The right panel displays the maximum value of $H$ with respect to $r$, specifically when
$p$ is fixed at 0.5.}
	\label{Fig.3}
\end{figure}

Based on the above result, we obtain the following theorem.
\begin{theorem}
The special X-states in Eq. (\ref{eq:EfGHZType}) are genuine tripartite nonlocal if and only if the corresponding states in Eq. (\ref{eq:EfSystemState}) have nonvanishing coherence.
\end{theorem}

\begin{proof}
On the one hand, according to the definition of coherence measure introduced in Eq. (\ref{eq:Cfmix}), $|r|>0$ implies that the coherence of the real system state in Eq. (\ref{eq:EfSystemState}) is nonvanishing.
On the other hand, $H>0$ for the state Eq. (\ref{eq:EfGHZType}) means that the state Eq. (\ref{eq:EfGHZType}) violates the Bell inequality Eq. (\ref{eq:hardytest}).
With the above result depicted in Fig. \ref{Fig.3}, we complete the proof.
\end{proof}

Genuine multipartite steering (GMS) is a concept introduced in \cite{QM}, representing the insufficiency of the directed hybrid local hidden variable and local hidden state (LHV-LHS) model.
For example, for three particles $I=\{1,2,3\},$ we have three bipartitions $\{1|23,2|13,3|12\}$ and the set $\mathcal{P}=\{\{1\}, \{2\}, \{3\}, \{1,2\}, \{1,3\},\{2,3\}\},$ thus the genuine multipartite steering is just the failure of
the hybrid LHV-LHS model
\begin{equation}
\begin{aligned}
&&p(c_1c_2c_3|M_1M_2M_3)~~~~~~~~~~~~~~~~~~~~~~~~~\\ \nonumber
&&=\sum\limits_{\alpha\in\mathcal{P}}\sum\limits_{\lambda}p(\lambda)p(c_{\alpha}|M_{\alpha},\lambda)p_Q(c_{\bar{\alpha}}|M_{\bar{\alpha}},\lambda)
\end{aligned}
\end{equation}
with $p_Q(C_{\bar{\alpha}}|M_{\bar{\alpha}},\lambda)=\rm{Tr}[\tau_{\lambda}\textit{M}_{\textit{x}_{\bar{\alpha}}}^{\textit{C}_{\alpha}}]$ and $\tau_{\alpha}$ the local hidden state.
As demonstrated in \cite{ZRS}, the authors have successfully established that GMS occupies an intermediate position between genuine multipartite entanglement (GME) and genuine multipartite nonlocality (GMNL).

For the special X-states given in Eq. (\ref{eq:EfGHZType}), their GME is equivalent to the coherence of the corresponding system states in Eq. (\ref{eq:EfSystemState}) according to Theorem 3.
With the result of Lemma 1, we obtain that any mixed state Eq. (\ref{eq:EfGHZType}) is genuine tripartite nonlocal iff its GME is nonzero. More specifically, the GME and GMNL is equivalent for the special X-states in Eq. (\ref{eq:EfGHZType}).
Moreover, leveraging the findings from Theorem 2 in \cite{ZRS}, we establish that for the special X-states, as defined by Eq. (\ref{eq:EfGHZType}), GME, GMS and GMNL are indeed equivalent. Consequently, it can be inferred that a state, as described by Eq. (\ref{eq:EfGHZType}), is genuinely tripartite steerable if the corresponding system state in Eq. (\ref{eq:EfSystemState}) has nonvanishing coherence, and vice versa.
The aforementioned discussion can be encapsulated in the following theorem,
\begin{theorem}
The special X-states in Eq. (\ref{eq:EfGHZType}) are genuine tripartite steerable if and only if the corresponding system states in Eq. (\ref{eq:EfSystemState}) have nonvanishing coherence.
\end{theorem}

The preceding two theorems suggest that coherence can be operationally transformed into GMS and GMNL within the framework described in Section \Rmnum{3}. Thus  the GME, GMS, GMNL and coherence are equivalent for mixed states in Eq. (\ref{eq:EfGHZType}), which is shown in Fig. \ref{Fig.4}.

\begin{figure}
	\centering
	\includegraphics[width=0.7\linewidth]{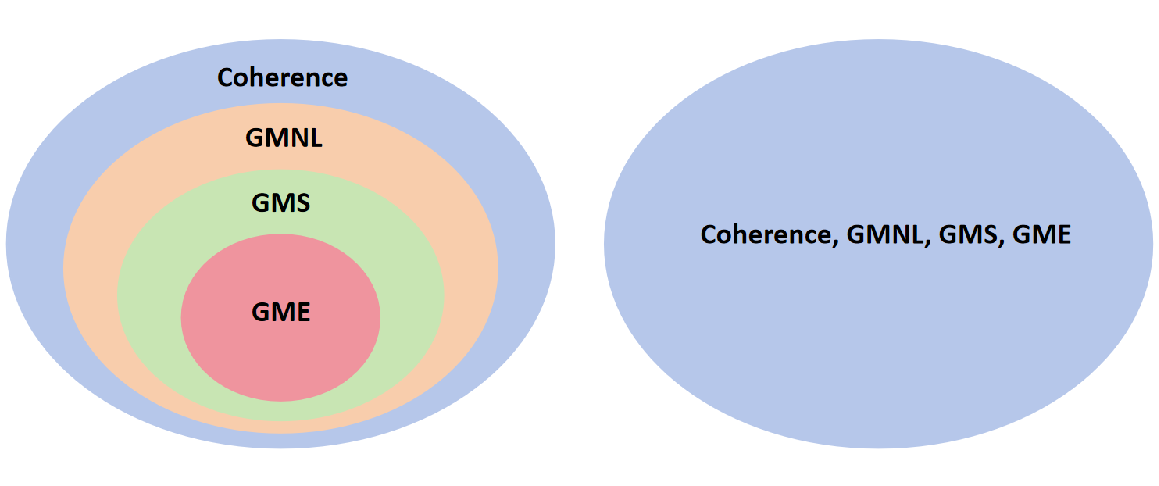}
	\caption{The interconnection among coherence, GMNL, GMS and GME for all quantum states (the states in Eq. (\ref{eq:EfGHZType})) is illustrated by the left (right) figure.}
	\label{Fig.4}
\end{figure}

\section{Summary}
We have constructed two types of genuine multipartite entanglement measures and have established the equivalence between these two types of genuine multipartite entanglement measures and coherence measures. Any coherence measure defined utilizing  symmetric concave functions for pure states with extension to mixed states by convex roof is equal to the genuine multipartite entanglement generated by a unitary incoherent operation on the coupling system including the initial system and an incoherent ancilla, which indicates that genuine multipartite entanglement can be generated from the coherence of initial state with unitary incoherent operation. Furthermore, we have investigated the operational connection between coherence and other quantum correlations such as genuine multipartite steering and genuine multipartite nonlocality. Any qubit quantum states can be transformed into a class of special X-states, via the unitary incoherent operation.
For these three qubit states, genuine multipartite nonlocality, genuine multipartite steering and genuine multipartite entanglement are equivalent. Moreover, these genuine multipartite correlations exist if and only if the coherence of the corresponding system state is present.

\section{Acknowledgments}
Zhihao Ma thanks for discussion with Qiongyi He, Huangjun Zhu and Zhengda Li. Zhihao Ma acknowledge the Fundamental Research Funds for the Central Universities, and Ming Li acknowledge the Fundamental Research Funds for the Central Universities (No. 22CX03005A). This work is supported by National Natural Science Foundation(NSFC) of China (Grant No. 12271325, 12071179, 12125402, 12371132,12075159,12171044); Shandong Provincial Natural Science Foundation for Quantum Science No. ZR2021LLZ002; Shenzhen Institute for Quantum Science and Engineering, Southern University of Science and Technology (Grant Nos. SIQSE202005); The specific research fund of the Innovation Platform for Academicians of Hainan Province.

\section{APPENDIX A: Proof Of Theorem 1}

$\textit{Theorem}$ 1.
For any  $f\in\scf$, $E^{f}_\gme$ and $G^{f}_\gme$ are both entanglement monotones. If in addition $f\in\scfp$, then $E^{f}_\gme$ and $G^{f}_\gme$ are both GME measures.

\begin{proof} 
To prove  $E^{f}_\gme$ and $G^{f}_\gme$ are both entanglement monotones for any $f\in\scf$, we need to prove the two quantities satisfy conditions (C2-C4). 
By the nature of the convex roof construction, $E^{f}_\gme(\rho)$ and $G^{f}_\gme(\rho)$ are convex. In addition, to show that both of them do not increase under LOCC, we may assume that the initial state $\rho$ is pure. Then $E^{f}_\gme(\rho)=E^f_{\gamma_0} (\rho)$ for some partition $\gamma_0$. Suppose a given LOCC transforms  $\rho$ to $\rho_j$ with probability $p_j$. Then
\begin{equation*}
	\sum_j p_j E^{f}_\gme(\rho_j)\leq \sum_j p_j E^{f}_{\gamma_0}(\rho_j)\leq E^{f}_{\gamma_0}(\rho)= E^{f}_\gme(\rho),
\end{equation*}
where the first equation follows from the definition of $E^{f}_\gme$ and the second from the monotonicity of $E^{f}_{\gamma_0}$ under selective operations \cite{HZ,GV}.
In addition, by using the concavity of geometric mean function and Mahler's inequality, we can get
\begin{equation*}
	\begin{aligned}
		G^{f}_\gme(\rho)&:=[\prod_{\gamma}E^{f}_{\gamma}(\rho)]^{\frac{1}{c(\alpha)}}\geq [\prod_{\gamma}\sum_{j}p_{j}E^{f}_{\gamma}(\rho_{j})]^{\frac{1}{c(\alpha)}}\\
		&\geq\sum_{j}p_{j}[\prod_{\gamma}E^{f}_{\gamma}(\rho_{j})]^{\frac{1}{c(\alpha)}}=\sum_{j}p_{j}G^{f}_\gme(\rho_{j}).
	\end{aligned}
\end{equation*}
Therefore, $E^{f}_\gme$ and $G^{f}_\gme$ both satisfy conditions (C2) and (C4); condition (C3) is a consequence of the two conditions.
	
Due to $f\in\scfp$, it is straightforward to verify that  $E^{f}_\gme$ and $G^{f}_\gme$ vanish when $\rho$ is biseparable and both are positive when $\rho$ is genuine multipartite entangled. Therefore, $E^{f}_\gme$ and $G^{f}_\gme$ are both GME measures.
\end{proof}

\section{APPENDIX B: Proof Of Theorem 2}
$\textit{Theorem}$ 2.
For an arbitrary $N$-partite mixed state $\rho$, one can prove the following inequality,
\begin{equation}\label{}
  G_{\mathrm{GME}}^f(\rho)\geq E_{\mathrm{GME}}^f(\rho).
\end{equation}

\begin{proof}
We can prove  $G_{\mathrm{GME}}^f(\rho)\geq E_{\mathrm{GME}}^f(\rho)$ by proving that for any state $|\psi\rangle_i$ in the optimal pure state decompositions of $\rho$ with respect to $G_{\mathrm{GME}}(\rho)$, we have $G_{\mathrm{GME}}^f(|\psi_i\rangle)\geq E_{\mathrm{GME}}^f(|\psi_i\rangle)$.

Firstly suppose that we have found the optimal pure state ensembles $\{q_i, |\psi_i\rangle\}$ for $G_{\mathrm{GME}}^f(\rho)$, i.e.,
\begin{eqnarray}
	\begin{aligned}
		G_{\mathrm{GME}}^f(\rho)&=\min_{\{p_j,|\phi_j\rangle\}}\sum_j p_j G_{\mathrm{GME}}^f(|\phi_j\rangle)\\
		&=\sum_i q_i G_{\mathrm{GME}}^f(|\psi_i\rangle).
	\end{aligned}
\end{eqnarray}
Now, for each pure state $|\psi_i\rangle$, one has
\begin{equation}\label{eq:Efgm}
	\begin{aligned}
		G_{\mathrm{GME}}^f(|\psi_i\rangle)&=\sqrt[c(\alpha)]{\Pi_\gamma E_\gamma^f(|\psi_i\rangle)}\geq \min_\gamma E_\gamma^f(|\psi_i\rangle)\\
		&=E_{\mathrm{GME}}^f(|\psi_i\rangle),
	\end{aligned}
\end{equation}
The above inequality holds because for arbitrary $\gamma$ we have $E_\gamma^f(|\psi_i\rangle)\geq\min_\gamma E_\gamma^f(|\psi_i\rangle)\geq0$. 
 
Then suppose that we have found the optimal pure state ensembles $\{q'_{k}, |\psi'_k\rangle\}$ for $E_{\mathrm{GME}}^f(\rho)$, i.e.,
\begin{eqnarray}
		\begin{aligned}
		E_{\mathrm{GME}}^f(\rho)&=\min_{\{p_j,|\phi_j\rangle\}}\sum_j p_j E_{\mathrm{GME}}^f(|\phi_j\rangle)\\
		&=\sum_k q'_k E_{\mathrm{GME}}^f(|\psi'_k\rangle).
	\end{aligned}
\end{eqnarray}
Therefore,
\begin{eqnarray}\label{eq:Efgmt}
	\begin{aligned}
		G_{\mathrm{GME}}^f(\rho)&=\sum_i q_i G_{\mathrm{GME}}^f(|\psi_i\rangle)\geq \sum_i q_i E_{\mathrm{GME}}^f(|\psi_i\rangle)\\
		&\geq  \sum_k q'_k E_{\mathrm{GME}}^f(|\psi'_k\rangle)=E_{\mathrm{GME}}^f(\rho),
	\end{aligned}
\end{eqnarray}
where we have used Eq.(\ref{eq:Efgm}) for the first inequality, and the fact that $\min_{\{p_j,|\phi_j\rangle\}}\sum_j p_j E_{\mathrm{GME}}^f(|\phi_j\rangle)= \sum_k q'_k E_{\mathrm{GME}}^f(|\psi'_k\rangle)$ for the second inequality.
\end{proof}

\section{APPENDIX C: Proof Of Theorem 3}
$\textit{Theorem}$ 3.
For any $f\in\scfp$, we have
\begin{equation*}
E^{f}_{GME}(\rho')=G^{f}_{GME}(\rho')=C^{f}(\rho).
\end{equation*}

\begin{proof}
Firstly, we prove the equality holds for pure state, then we can prove the equality for any quantum state by using the definitions of $E^{f}_{GME}(\rho')$ and $G^{f}_{GME}(\rho')$, and one-to-one mapping between the ensembles of $\rho$ and that of $\rho'$.

When $\rho$ is pure, say $\rho=|\psi\rangle\langle\psi|$ with $|\psi\rangle=\sum_{j}c_{j}|j\rangle$. We can easily get $\rho'=|\psi'\rangle\langle\psi'|$ with $|\psi'\rangle=\sum_{j}c_{j}|j_{A_{0}}j_{A_{1}}...j_{A_{N-1}}\rangle$, which is the special X-state, where ${A_{i}},i=1,2,...,N-1$ corresponds to the ancilla system. For any bipartition $\gamma|\bar{\gamma}$ of $\rho'$, the schmidt vector $\lambda_{\gamma}(\rho')=(|c_{0}|^{2}, |c_{1}|^{2},... |c_{d-1}|^{2})^{\rmT}$, which is equal to the coherence vector of $\rho$, thus for any $f\in\textit{F}_{sc}$, $f(\lambda_{\gamma}(\rho))=f(\mu(\ket{\psi}))$, i.e., $E^{f}_{GME}(\rho')=G^{f}_{GME}(\rho')=C^{f}(\rho)$.
	
When $\rho$ is a mixed state. Suppose $\rho=\sum_{i}p_{i}|\varphi_{i}\rangle\langle\varphi_{i}|$ is the optimal ensemble decomposition of $\rho$, i.e., $C^{f}(\rho)=\sum_{i}p_{i}C^{f}(|\varphi_{i}\rangle)$. Note that the isometry $U$ establishes a one-to-one mapping between the ensembles of $\rho$ and that of $\rho'$, i.e.,
\begin{equation*}
	|\Phi_{i}\rangle=U(|\varphi_{i}\rangle\otimes|0\rangle^{\otimes N-1}),
\end{equation*}
where $\rho'=\sum_{i}p_{i}|\Phi_{i}\rangle\langle\Phi_{i}|$. So we get
\begin{equation*}
	C^{f}(\rho)=\sum_{i}p_{i}E^{f}_{GME}(|\Phi_{i}\rangle)=\sum_{i}p_{i}G^{f}_{GME}(|\Phi_{i}\rangle).
\end{equation*}
According to the convex of entangled measure, we have
\begin{equation}\label{eq:EfleqCf}
	E^{f}_{GME}(\rho')\leq \sum_{i}p_{i}E^{f}_{GME}(|\Phi_{i}\rangle)=C^{f}(\rho),
\end{equation}
\begin{equation}\label{eq:GfleqCf}
	G^{f}_{GME}(\rho')\leq \sum_{i}p_{i}G^{f}_{GME}(|\Phi_{i}\rangle)=C^{f}(\rho).
\end{equation}
Similarly, we suppose $\rho'=\sum_{i}q_{i}|\Psi_{i}\rangle\langle\Psi_{i}|$ is the optimal ensemble decomposition of $\rho'$, i.e., $E^{f}_{GME}(\rho')=\sum_{i}q_{i}E^{f}_{GME}(|\Psi_{i}\rangle)=\sum_{i}q_{i}G^{f}_{GME}(|\Psi_{i}\rangle)=G^{f}_{GME}(\rho')$. Then we can get an ensemble decomposition of $\rho$ according to the isometry $U$, i.e.,
\begin{equation*}
	|\Psi_{i}\rangle=U(|\phi_{i}\rangle\otimes|0\rangle^{\otimes N-1}),
\end{equation*}
where $\rho=\sum_{i}q_{i}|\phi_{i}\rangle\langle\phi_{i}|$.
According to the convex of coherence measure, we have
\begin{equation}\label{eq:CfleqEfGf}
	\begin{split}
		C^{f}(\rho)\leq \sum_{i}q_{i}C^{f}(|\phi_{i}\rangle)=\sum_{i}q_{i}E^{f}(|\Psi_{i}\rangle)\\
		=E^{f}_{GME}(\rho')=G^{f}_{GME}(\rho')\\.
	\end{split}
\end{equation}
Thus, the proof is completed.
\end{proof}

\section{APPENDIX D: Calculation Of Maximal Value Of \textit{H} For the Special X-States}
In order to find the observables $M_{x_i}$ for each particle $i$ that maximize $H$ in the left hand of Eq. (\ref{eq:EfHardy2}), we let

\begin{equation*}
	\begin{aligned}
		M_{x_1=0}=|a_{1}\rangle\langle a_{1}|-|\bar{a}_{1}\rangle\langle \bar{a}_{1}|,
		M_{x_1=1}=|b_{1}\rangle\langle b_{1}|-|\bar{b}_{1}\rangle\langle \bar{b}_{1}|		
	\end{aligned}
\end{equation*}
and
\begin{equation*}
	\begin{aligned}
		M_{x_k=0}=|a\rangle\langle a|-|\bar{a}\rangle\langle \bar{a}|,
		M_{x_k=1}=|b\rangle\langle b|-|\bar{b}\rangle\langle \bar{b}|		
	\end{aligned}
\end{equation*}
for $k=2,3$, where

\begin{equation*}
	\begin{aligned}
		&|a_{1}\rangle=\cos\theta_{1}|0\rangle+\sin\theta_{1}|1\rangle, |b_{1}\rangle=\cos\theta_{2}|0\rangle+\sin\theta_{2}|1\rangle, \\
		&|a\rangle=\cos\theta_{3}|0\rangle+\sin\theta_{3}|1\rangle, |b\rangle=\cos\theta_{4}|0\rangle+\sin\theta_{4}|1\rangle, \\	
		&|\bar{a}_{1}\rangle=\sin\theta_{1}|0\rangle-\cos\theta_{1}|1\rangle, |\bar{b}_{1}\rangle=\sin\theta_{2}|0\rangle-\cos\theta_{2}|1\rangle,\\
		&|\bar{a}\rangle=\sin\theta_{3}|0\rangle-\cos\theta_{3}|1\rangle, |\bar{b}\rangle=\sin\theta_{4}|0\rangle-\cos\theta_{4}|1\rangle,\\
	\end{aligned}
\end{equation*}
where $\theta_{i}\in[0,\pi]$ for $i=1,2,3,4$.
Then one can calculate
\begin{eqnarray}\label{eq:EfHardyValue}
	\begin{aligned}
		H=&(p-1)\sin^{2}\theta_{3}[2\cos^{2}\theta_{2}\cos^{2}\theta_{4}+(\sin^{2}\theta_{2} \\
		&-\sin^{2}\theta_{1})\sin^{2}\theta_{3} +2\sin^{2}\theta_{1}\sin^{2}\theta_{4}]
		+p\cos^{2}\theta_{3} \\
		&[\cos^{2}\theta_{1}
		(\cos^{2}\theta_{3}-2\cos^{2}\theta_{4})
		-\cos^{2}\theta_{2}\cos^{2}\theta_{3} \\
		&-2\sin^{2}\theta_{2}\sin^{2}\theta_{4}]+r\sin2\theta_{3}
		[\cos\theta_{1}\sin\theta_{1}
		\\
		&(\cos\theta_{3}\sin\theta_{3}-\sin2\theta_{4})-\cos\theta_{2}\sin\theta_{2} \\
		&(\cos\theta_{3}\sin\theta_{3}+\sin2\theta_{4})].	
	\end{aligned}
\end{eqnarray}

Numerically, one can find the maximal value of Eq. (\ref{eq:EfHardyValue}) by using the attached Mathematica function FindMaximum for all eligible $p$ and $r$ to guarentee $\rho'$ in Eq. (\ref{eq:EfGHZType}) is a quantum state (that is, $\rho'$ is semi-positive and the trace of $\rho'$ is one), which means that $p$ and $r$ satisfy $0\leq p\leq 0.5$ and $0\leq |r|\leq\sqrt{p(1-p)}$. For example, when $p=0.5$ and $r=0.4$, the code for calculating the result of Eq. (\ref{eq:EfHardyValue}) is
\begin{equation*}
	\begin{aligned}
		&p=0.5; r=0.4; FindMaximum[ \\
		&pCos[\theta_{3}]^{2}(Cos[\theta_{1}]^{2}(Cos[\theta_{3}]^{2}-2Cos[\theta_{4}]^{2}) \\
		&-Cos[\theta_{2}]^{2}Cos[\theta_{3}]^{2}-2Sin[\theta_{2}]^{2}Sin[\theta_{4}]^{2}) \\
		&+(p-1)Sin[\theta_{3}]^{2}(Sin[\theta_{3}]^{2}(Sin[\theta_{2}]^{2}-Sin[\theta_{1}]^{2}) \\
		&+2Sin[\theta_{1}]^{2}Sin[\theta_{4}]^{2}+2Cos[\theta_{2}]^{2}Cos[\theta_{4}]^{2}) \\
		&+2rSin[\theta_{3}]Cos[\theta_{3}](Sin[\theta_{1}] Cos[\theta_{1}](Sin[\theta_{3}]
		Cos[\theta_{3}] \\
		&-2Sin[\theta_{4}]Cos[\theta_{4}])-Sin[\theta_{2}]Cos[\theta_{2}] (Sin[\theta_{3}]Cos[\theta_{3}] \\
		&+Sin[2\theta_{4}])), \\
		&\{\theta_{1}, \frac{\pi}{2}\}, \{\theta_{2}, \frac{\pi}{2}\}, \{\theta_{3}, \frac{3\pi}{4}\}, \{\theta_{4}, 0\}]
	\end{aligned}
\end{equation*}

With the calculation using Mathematica, we can obtain that the maximal value of $H$ for any state described in Eq. (\ref{eq:EfGHZType}) is positive if and only if $|r|>0$, as depicted in Fig. \ref{Fig.3}.

\end{document}